\documentclass[runningheads]{llncs}
\usepackage[T1]{fontenc}
\usepackage{graphicx}
\usepackage{hyperref}
\usepackage{color}

\usepackage{algorithm}
\usepackage{algorithmic}
\usepackage{comment}
\usepackage{amssymb,amsmath}
\usepackage[usenames,dvipsnames,svgnames,table]{xcolor}

\newtheorem{observation}{Observation}

\newcommand{\tj}[1]{\textcolor{red}{#1}}

\begin{document}
\title{Friend- and Enemy-oriented Hedonic Games with Strangers}
%
%

\author{TJ Schlueter\inst{1}\orcidID{0000-0002-8365-4186} \and
Makoto Yokoo\inst{2}\orcidID{0000-0003-4929-396X}}

%
%
\institute{Union College, Schenectady NY 12308, USA \email{schluett@union.edu} \and
Kyushu University, Fukuoka, Fukuoka Japan
\email{yokoo@inf.kyushu-u.ac.jp}}
%
\maketitle              
\begin{abstract}
We introduce friend- and enemy-oriented hedonic games with strangers (FOHGS and EOHGS respectively), two classes of hedonic games wherein agents are classified as friends, enemies, or strangers under the assumption that strangers will become either friends or enemies ex post facto.
For several notions of stability in FOHGS and EOHGS, we characterize the hardness of verification for possible and necessary stability.
We characterize the hardness of deciding whether possibly and necessarily X stable partitions exist for a given stability notion X.
We prove that necessarily internally stable partitions always exist and provide sufficient conditions for necessary contractual individual stability.
\end{abstract}
\section{Introduction}




Many real-life situations such as political parties, recreational sports teams, and research groups require individuals to form groups to accomplish their goals.
In such situations, individuals usually have preferences over who they are grouped with.
Hedonic games, formally hedonic coalition formation games, model the partitioning of a pool of individuals (agents) into such groups (coalitions), but are intractable in the general case.
Friend- and enemy-oriented hedonic games
(FOHG and EOHG)
are hedonic games which provide several tractability benefits by assuming that each agent views each other agent as either friend or enemy and that each agent's preferences are determined by the number of friends and enemies in their coalition \cite{dimitrov2006simple}.
We introduce friend-oriented and enemy-oriented hedonic games with \textit{strangers} (FOHGS and EOHGS) as a relaxation of the FOHG and EOHG models so agents do not need to classify \textit{all} other agents as friends or enemies a priori,
but that that strangers will become either friends or enemies after becoming acquainted.

Hedonic game research often focuses on partitions that satisfy certain conditions, such as various notions of \textit{stability}, where agents will not change coalitions ex post facto, given certain restrictions and assumptions.
Due to the uncertainty caused by strangers in our models, existing notions of stability do not fit FOHGS and EOHGS properly.

Our proposed notion of strangers is similar to unknowns in voting theory, which can make the outcome of an election unclear.
\cite{konczak2005voting} proposed possible and necessary winners for voting to handle unknowns. 
Informally, a possible winner may win an election if enough unknown votes go in their favor, and a necessary winner wins even if all unknown votes go against them.
All necessary winners are possible winners, but possible winners may not be necessary winners.

We adapt the notions of possible and necessary winners in voting to possibly and necessarily stable partitions for FOHGS and EOHGS.
A partition is \textit{possibly X stable} for stability notion X if there exists a resolution of stranger relationships that makes the partition X stable, whereas it is \textit{necessarily X stable} if it will be X stable regardless of how stranger relationships are resolved.
Similar to voting, all necessarily stable partitions are possibly stable, but not the other way around.

Our first results investigate the complexity of possible and necessary stability verification.
We make some straightforward
observations regarding possible stability existence.
Then, we provide a mix of positive and negative results characterizing necessary stability existence.
Most positive results are for symmetric FOHGS and EOHGS, while most results for asymmetric FOHGS and EOHGS are negative.

There are exponentially many ways for the strangers in a given FOHGS or EOHGS instance to be resolved.
If, in the worst case, verification of some variation of possible or necessary stability requires all resolutions of strangers to be checked, then that verificaiton problem is intractable.
We prove that only 2 resolutions of strangers need to be checked, leading to several positive tractability results.
In FOHG and EOHG partitions satisfying several stability notions (e.g. core stability) are guaranteed to exist. Symmetric FOHG and EOHG provide additional existence guarantees.
While it is not surprising that some stability guarantees are lost when moving to a more general model, we find it surprising that even contractual individual stability cannot be guaranteed without some graph restrictions while necessary individual stability and necessary core stability cannot be guaranteed even with rather strong graph restrictions.

\section{Related Work}

The study of hedonic games was initiated by Dreze and Greenberg; later codified by Banerjee et al. and Bogomolnaia and Jackson \cite{DrezeGreenbergHC1980,banerjee2001core,bogomolnaia2002stability}. 
Peters and Elkind  provide an overview of complexity results in several classes of hedonic games \cite{peters2015simple}. 
Dimtrov et al. were the first to propose FOHG and EOHG \cite{dimitrov2006simple},
which our work extends.


Dimtrov et al. showed that a partition of strongly connected components is always strictly core stable (SCS) for FOHG and such partitions can be computed in polynomial time (P-time) for any FOHG \cite{dimitrov2006simple}; they also proved that a core stable partition always exists for EOHG \cite{dimitrov2006simple}.
Sung and Dimitrov later proved that verifying whether a partition is SCS or CS for EOHG is coNP-complete \cite{SUNG2007155}.
Rey later proved the SCS existence problem for EOHG to be DP-hard. \cite{rey2016beyond}.
Chen et al. proved that the CS and SCS verification problems are coNP-complete for FOHG, EOHG \cite{chenhedonic23}.


Bogomolnaia and Jacksomn proved that Nash stable (NS) partitions always exist for symmetric  additively separable hedonic games (ASHG) \cite{bogomolnaia2002stability}; thus
individually stable (IS) and contractual individually stable (CIS) partitions also always exist.
Since FOHG and EOHG are subclasses of symmetric ASHG, Bogomolnaia and Jackson's findings mean that NS, IS, and CIS partitions always exist for FOHG and EOHG.
Sung and Dimitrov proved that the NS and IS existence problems for asymmetric ASHG are NP-complete \cite{sung2010computational}; 
however, these results don't necessarily hold for FOHG and EOHG.
Brandt et al. \cite{brandt2022single} showed that NS existence is NP-complete for any subclass of asymmetric ASHG with only one positive and one negative utility value \cite{brandt2022single}.
They also proved that IS partitions always exist and can be found in P-time for asymmetric ASHG with at most one non-negative utility value \cite{brandt2022single}.
Chen et al. extended the findings of Brandt et al. by showing that the Nash existence problem remains NP-complete, even when the friendship graph is planar \cite{chenhedonic23}.
Dimtrov et al. have shown that a partition of strongly connected components is always SCS for FOHG.  \cite{dimitrov2006simple}. However, the authors made no claims that such partitions are the \textit{only} SCS partitions.


Many extensions have been proposed for FOHG and EOHG since their introduction.
Ota et al. proposed FOHG/EOHG with \textit{neutrals}, which have  zero impact on utility and showed that, after the introduction of \textit{neutrals}, the CS and SCS existence problems are $\Sigma_2^P$-complete \cite{ota2017core}.
Our work differs in that strangers have a non-zero impact on utility.
Barrot et al. proposed an extension where agents can be categorized as friend, enemy, or \textit{unknown} with unknowns having either an $\epsilon$ positive or negative impact on utility \cite{barrot2019unknown}.
In our work strangers have a larger than $\epsilon$ impact on utility and agents do not necessarily resolve all their stranger relations.
Nguyen et al. proposed a variation of the FOHG model where agents altruistically consider the utility that their friends derive from their coalition \cite{nguyen2016altruistic}.
This model was extended by Schlueter and Goldsmith \cite{schlueter2020super} to consider all agents in the coalition, and by Kerkmann and Rothe to the more general space of coalition formation games \cite{schlueter2020super,kerkmann2020altruism}.
We assume in our models that agents act exclusively in their own self-interest.

Skibski et al. proposed signed graph games, a class of games where edges between agents can be positive or negative, and feasible coalitions must not contain any pair of agents connected by a negative edge, and investigate various semivalues and the Owen value for such games \cite{skibski2020signed}.
These games have transferable utility, whereas all hedonic games have strictly nontransferable utility.

Lang et al. proposed an adaptation of possible and necessary stability for hedonic games with ordinal preferences and thresholds \cite{lang2015representing}; 
games where each agent partitions other agents into friends, enemies, and neutrals, then provides a weak ordering of friends and enemies.
In their model, possible/necessary stability is based on the Bossong–Schweigert extension of agents' preferences \cite{lang2015representing};
if an extension exists such that a partition is X stable, then it is possibly X stable, and if it remains X stable for all possible preference extensions, then it is necessarily X stable.
Kerkmann et al. later expanded this work with results for additional stability notions and for some notions of optimality \cite{kerkmann2020hedonic}.
Friends and enemies are not ranked in our model and the model proposed by Lang et al. does not capture the notion of strangers that ours does \cite{lang2015representing}.
Further, their proposed notions of possible and necessary stability are incompatible with our model and vice-versa.


\section{Preliminaries}

Given a game with a set of agents $N$, a \textbf{coalition} is some set of agents $C\subseteq N$ and a \textbf{partition} is a set of coalitions $\gamma$ such that $\bigcup_{C\in \gamma} C = N$ and $\forall C\neq C'\in \gamma$, $C\cap C'= \emptyset$. Next, we define classes of games relevant to our work.


\begin{definition}
\label{def:HG}
\cite{banerjee2001core,bogomolnaia2002stability} \textbf{Hedonic games} are coalition formation games with nontransferable utility where agents' preferences can be expressed as a weak ordering of all coalitions which contain them.
A solution for a given hedonic game is a partition $\gamma$ of the game's agents.
Let $\Gamma$ denote the set of all partitions.
Each agent has preferences over all partitions $\gamma \in \Gamma$ based solely upon their assigned coalition, denoted by $\gamma(i)\in \gamma$.
Given agent $i\in N$ and partitions $\gamma, \gamma' \in \Gamma$, $\gamma \succ_i\gamma'$ if $i$ strictly prefers $\gamma$ to $\gamma'$ and $\gamma\succeq_i\gamma'$ if $i$ weakly prefers $\gamma$ to $\gamma'$.
\end{definition}

There are myriad subclasses of hedonic games, so we only define subclasses relevant to our work.

\begin{definition}
\label{def:ASHG}
\cite{banerjee2001core} \textbf{Additively Separable Hedonic Games (ASHG)} are hedonic games where all agents $i\in N$ assign utility values to other agents $j\in N\setminus\{i\}$, denoted by $v_i(j)$.
An agent's preference come from sum of the utility values assigned to other agents in their coalition.
For some $i\in N$ $u_i(\gamma)=\sum_{j\in\gamma(i)\setminus\{i\}}v_i(j)$ denotes the utility $i$ gets from coalition $\gamma(i)\in\gamma$; $u_i(C)$ can be used for any coalition $C\subseteq N$.
Given agent $i$ and $\gamma,\gamma'\in \Gamma$, $\gamma\succ_i\gamma'$ if $u_i(\gamma) > u_i(\gamma')$ and $\gamma\succeq_i\gamma'$ if $u_i(\gamma) \geq u_i(\gamma')$.
\end{definition}

A common restriction for ASHG and their derivations is symmetry, which we define below.

\begin{definition}
An ASHG is \textbf{symmetric} if and only if:
$\forall i,j\in N,\mbox{ } v_i(j) = v_j(i).$
\end{definition}

ASHG instances that are not symmetric are referred to as \textbf{asymmetric} or \textbf{general-case}.
We now introduce two restricted subclasses of ASHG.

\begin{definition}
\label{def:FEOHG}
\cite{dimitrov2006simple} \textbf{Friend-oriented Hedonic Games (FOHG)} are ASHG where all agents $i\in N$ categorize other agents $j\in N\setminus\{i\}$, as either friends or enemies with friends valued at $|N|$ and enemies $-1$.
\textbf{Enemy-oriented Hedonic Games (EOHG)} analogous, but with friends valued at $1$ and enemies $-|N|$.

\end{definition}


In FOHG agents strictly prefer any coalition with more friends over a coalition with fewer friends.
In EOHG agents strictly prefer any coalition with less enemies over a coalition with more enemies.
Since a single enemy guarantees that an agent will derive negative utility, all individually rational (IR) partitions are composed solely of coalitions of cliques on the friendship graph for EOHG.

FOHG and EOHG are often represented with unweighted graphs where agents are nodes and an edge between two agents indicates friendship.
Given a FOHG or EOHG, let $F$ denote the set of all friendship edges and $E$ denote the set of all pairs $(i,j)\notin F$; let $F_i$ and $E_i$ denote the sets of friends and enemies respectively for agent $i\in N$.
Given a coalition $C$, the utility agent $i\in C$ derives from $C$ can be computed as $\sum_{\forall j\in C\setminus\{i\}}v_i(j)$ or as $|N|\cdot|F_i\cap C| - |E_i\cap C|$ for FOHG and $|F_i\cap C| - |N|\cdot|E_i\cap C|$ for EOHG.
Next, we define notions of stability that focus on the deviation of individual agents.




\begin{definition}
\label{def:individual_rationality}
\cite{rothe2015economics_Chapter3}
A partition $\gamma$ is \textbf{individually rational (IR)} if $\forall i\in N$, $\gamma(i) \succeq_i \{i\}$.
\end{definition}

IR is a common baseline of stability for hedonic games as it ensures that no agent becomes worse-off by participating.

\begin{definition} \label{def:Nash_stab}
\cite{bogomolnaia2002stability} A partition $\gamma$ is \textbf{Nash stable (NS)} if $\nexists i\in N,$ $C\in \gamma \cup \{\emptyset\}$ such that $i\notin C$ and $C\cup\{i\} \succ_i \gamma(i)$.
\end{definition}

NS guarantees that unilateral deviation benefits no one.

\begin{definition} \label{def:indiv_stab}
\cite{bogomolnaia2002stability} A partition $\gamma$ is \textbf{individually stable (IS)} if $\nexists i\in N,$ $C\in \gamma \cup \{\emptyset\}$ such that $i\notin C$; $C\cup\{i\} \succ_i \gamma(i)$; and $\forall j\in C$, $C\cup\{i\} \succeq_j C$.
\end{definition}

IS builds on NS by requiring permission from existing members before an agent can join a coalition; 
NS and IS guarantee IR.

\begin{definition} \label{def:contract_indiv_stab}
\cite{bogomolnaia2002stability}  A partition $\gamma$ is \textbf{contractually individually stable (CIS)} if $\nexists i\in N,$ $C\in \gamma \cup \{\emptyset\}$ such that $i\notin C$; $C\cup\{i\} \succ_i \gamma(i)$; $\forall j\in C$, $C\cup\{i\} \succeq_j C$; and $\forall k\in \gamma(i)\setminus\{i\},$ $\gamma(i)\setminus\{i\}\succeq_k \gamma(i)$.
\end{definition}

CIS builds on IS by requiring permission for an agent to leave their current coalition.
CIS does not guarantee IR.
We now define stability notions where multiple agents can deviate simultaneously.



\begin{definition} \label{def:core_stab}
\cite{bogomolnaia2002stability} Coalition $C$ \textbf{blocks} a partition $\gamma$ if $\forall i \in C$, $C \succ_i \gamma(i)$.
Partition $\gamma$ is \textbf{core stable (CS)} if no non-empty coalition blocks it.
\end{definition}

A strengthening of the ideas in CS was proposed by \cite{dimitrov2006simple}. 

\begin{definition} \label{def:strict_core_stab}
\cite{dimitrov2006simple} Coalition $C$ \textbf{weakly blocks} partition $\gamma$ if $\forall i \in C$, $C \succeq_i \gamma(i)$ and $\exists j\in C$ such that $C \succ_j \gamma(j).$
Partition $\gamma$ is \textbf{strictly core stable (SCS)} if no coalition weakly blocks it.
\end{definition}

A final relevant stability notion was proposed by Schlueter and Goldsmith, built on work by Alcalde and Romero-Medina, Dimitrove et al., and Taywade et al \cite{alcalde2006coalition,dimitrov2006simple,schlueter2020internal,taywade2020decentralized}.

\begin{definition} \label{def:int_stab}
\cite{dimitrov2006simple,alcalde2006coalition} A coalition $C$ is \textbf{internally stable} if $\nexists D\subset C$ such that $\forall i\in D$, $D\succ_i C$.
\cite{schlueter2020internal} A partition $\gamma$ is \textbf{internally stable (INS)} if all coalitions $C\in \gamma$ are internally stable.
\end{definition}

\subsection{New Model Definitions}

We end our preliminaries with the introduction of the new model and related notions of possible and necessary stability.

\begin{definition}
\label{def:FEOHGS}
\textbf{Friend-oriented Hedonic Games with Strangers (FOHGS)} are an extension of FOHG where agents are categorized a priori as friend, enemy, or stranger.
Friends and enemies in FOHGS are assigned the same utility values as they are in FOHG, while strangers become either friends or enemies ex post facto.
Similar to FOHG and EOHG, we denote the set of friend edges with $F$, but we add another set $S$ to denote the set of all $(i,j)$ where $i\in N$ categorizes $j\in N:i\neq j$ as a stranger.
The set $E$ for FOHGS and EOHGS denotes the set of pairs $(i,j)$ such that $(i,j)\notin F \wedge (i,j)\notin S$.
We let $F_i$, $E_i$, $S_i$ denote the sets of friends, enemies, and strangers respectively for some agent $i\in N$.
For a FOHGS $G$, partition $\gamma$, and agent $i \in N$, $u_i(\gamma(i)) = |N|\cdot |F_i\cap \gamma(i)| - |E_i\cap \gamma(i)| + \sum_{j\in S_i\cap \gamma(i)} v_i(j)$ where 
\[ v_i(j) = 
\left\{ \begin{array}{lr}
    |N| & \text{if } (i,j) \text{ becomes friendship}\\
    -1 & \text{if } (i,j) \text{ becomes enmity}
\end{array} \right\} .\]

\textbf{Enemy-oriented Hedonic Games with Strangers} are defined analogously wrt EOHG.
Thus for an EOHGS $G$, partition $\gamma$, and agent $i \in N$, $u_i(\gamma(i)) = |F_i\cap \gamma(i)| - |N|\cdot |E_i\cap \gamma(i)| + \sum_{j\in S_i\cap \gamma(i)} v_i(j)$ where 
\[ v_i(j) = 
\left\{ \begin{array}{lr}
    1 & \text{if } (i,j) \text{ becomes friendship}\\
    -|N| & \text{if } (i,j) \text{ becomes enmity}
\end{array} \right\} .\]
A FOHGS or EOHGS instance $G$ can be defined by the tuple $(N, F, E, S)$.
\end{definition}

FOHGS and EOHGS can be symmetric or asymmetric.
We now define possible and necessary stability in the context of FOHGS and EOHGS.


\begin{definition}
Given a FOHGS or EOHGS instance $G$, a \textbf{resolution} $r$ of strangers in $G$ converts each stranger relation into friendship or enmity.
A partition $\gamma$ is \textbf{possibly X stable (P-X)} for stability notion X (e.g. NS) if $\gamma$ is X stable for some resolution of strangers in $G$.
Partition $\gamma$ is \textbf{necessarily X stable (N-X)} if $\gamma$ is X stable for all possible resolutions of strangers in $G$.
We use $u_i^r(C)$ to denote the utility $i$ gets from a coalition $C$ under resolution $r$.
\end{definition}



Before presenting our formal results, we propose a trivial case where all agents are strangers to each other, which applies to both FOHGS and EOHGS.
In this case, all possible partitions are possibly stable for many stability notions, notably for SCS, NS, and stability notions implied by them.
In contrast, no partition is necessarily stable for many stability notions - particularly for necessary CIS and INS and all stability notions that imply either of these notions (e.g. NS, SCS).
Our existence results are focused on nontrivial instances.

Next, we present results for FOHGS and EOHGS; first results for possible and necessary stability verification; confirming that a given partition satisfies a specific stability notion.
We then present informal possible stability existence results and technical necessary existence results; checking whether \textit{any} partition possibly or necessarily satisfies a specific stability notion.

\section{Verification of Possible/Necessary Stability}
Our first formal results relate to the tractability of verifying possible and necessary stability.
For any FOHGS or EOHGS instance $G = (N, F, E, S)$ there are $2^{|S|}$ resolutions of strangers.
Given a partition $\gamma$ for $G$, checking all possible resolutions of $S$ to verify possible or necessary X stability is intractable.
We show that it is possible to verify that coalition $C$ possibly blocks  partition $\gamma$ in P-time.




\begin{theorem}
\label{thm:FOHGS_OnlyExtremes}
A coalition $C$ can be verified to possibly block a partition $\gamma$ for a FOHGS or EOHGS instance $G$ in P-time.
\end{theorem}

The proof for Theorem \ref{thm:FOHGS_OnlyExtremes} relies on the following lemma and can be found in Appendix A.


\begin{lemma}
\label{lem:FOHGS_MinMax}
Given a FOHGS or EOHGS $G$ and coalition $C\subseteq N$, $\forall i\in C$, let $r^-$ and $r^+$ be the resolutions where all strangers become enemies and friends respectively.
The minimum utility $i$ can get from $C$ is $u_i^{r^-}(C)$, while $u_i^{r^+}(C)$ is the maximum.
\end{lemma}

Intuitively, it follows that converting all strangers into enemies produces the lowest utility, while converting strangers into friends maximizes utility; see Appendix A for the full proof of the lemma. Building upon Lemma \ref{lem:FOHGS_MinMax}, we move into the proof of Theorem \ref{thm:FOHGS_OnlyExtremes}.

\begin{proof}
Lemma 1 established that an agent's utility is maximized under resolution $r^+$ and minimized under resolution $r^-$. In order for coalition $C \subseteq N$ to possibly block partition $\gamma$, there must be a resolution $r$ such that $C$ provides each of its members $i\in C$ with more utility than they get from $\gamma(i)$ under the same resolution.
More than one such $r$ \textit{could} exist, but if they do, then we will find that $\forall i\in C$, $u_i^{r^+}(C) > u_i^{r^-}(\gamma(i))$.
But this finding alone is insufficient to prove that $C$ is a blocking coalition.
Suppose that $C$ was an existing coalition $C'\in \gamma$ and that $\forall i\in C$ (equivalently $C'$) $S_i\cap C\neq \emptyset$.
All agents in $C$ have at least one stranger, $C$ will appear to block $\gamma$ despite being a part of $\gamma$.
Similar problems arise when $C \subset C'\in \gamma$,
though in such cases, it may be possible that $C$ is disjoint from $C'$ on the friendship graph --- in which case, the agents of $C$ \textit{can} benefit by leaving $C'$.
Thus, there are three cases of note for how any $C\subseteq N$ relates to existing coalitions in $\gamma$.
\begin{enumerate}
    \item $C$ is part of $\gamma$, thus $C$ cannot block $\gamma$.
    \item $C\subset C'\in \gamma$. If $\forall i\in C$, $u_i^{r^-}(C) > u_i^{r^-}(C')$, then $C$ is disjoint from $C' \setminus C$ on the friendship graph.
    Thus, under $r^-$ or any other $r$ where all stranger edges between $C$ and $C'\setminus C$ become enmity, all $i\in C$ benefit if $C$ separates from $C'$.
    Thus $C$ is a possibly blocking coalition.
    \item $C$ is not an existing member of $\gamma$ or subset of some $C'\in \gamma$.
    In this case $\forall i\in C$, we must check whether $u_i^{r^+}(C\setminus \gamma(i)) > u_i^{r^-}(\gamma(i)\setminus C)$.
    Intuitively, stranger relations exclusively between $i$ and $C$ resolve positively, stranger relations exclusively between $i$ and $\gamma(i)$ resolve negatively, and other stranger relations are irrelevant.
    If this holds for all $i\in C$, then $C$ possibly blocks $\gamma$.
\end{enumerate}

In order to reach the condition 
$u_i^{r^+}(C\setminus \gamma(i)) > u_i^{r^-}(\gamma(i)\setminus C),$
we compute the following $\forall i\in C$:
$u_i^{r^+}(C\cap \gamma(i))$, $u_i^{r^+}(C\setminus \gamma(i))$, and $u_i^{r^-}(\gamma(i)\setminus C)$.
To compare $C$ and $\gamma(i)$, we start with the following inequality based on the above values: $u_i^{r^+}(C\cap \gamma(i)) + u_i^{r^+}(C\setminus \gamma(i)) > u_i^{r^+}(C\cap \gamma(i)) + u_i^{r^-}(\gamma(i)\setminus C)$.
We can subtract $u_i^{r^+}(C\cap \gamma(i))$ from from both sides to find the simplified inequality $u_i^{r^+}(C\setminus \gamma(i)) > u_i^{r^-}(\gamma(i)\setminus C)$.
Thus, $C$ can be verified to block $\gamma$ in P-time.

Analogous steps can be constructed for EOHGS by swapping the magnitude of positive and negative utility values.
Thus, it can be verified whether a coalition $C$ blocks a partition $\gamma$ of some EOHGS instance in polynomial time.
\qed
\end{proof}
We find similar results for verification of necessarily blocking coalitions.

\begin{proposition}
\label{prop:FEOHGS_PN_blocking}
A coalition can be verified to necessarily block a partition $\gamma$ for a FOHGS or EOHGS instance $G$ in P-time.
\end{proposition}

\begin{proof}
To check whether $C$ necessarily blocks $\gamma$, we can rely on a modified version of the steps used to check whether whether $C$ possibly blocks $\gamma$.
When $C$ is a proper subset of an existing coalition $C'$, we must check whether $C$ is necessarily better than $C'$.
For FOHGS this only holds when $C$ is necessarily disjoint from $C'$, equivalently that $\forall i\in C\mbox{, }|E_i\cap (C'\setminus C)| = |(C'\setminus C)|$.
For EOHGS, this holds whenever $C'$ has more known enemies than $C$, equivalently $\forall i\in C$, $|E_i\cap C'| > |E_i\cap C|$.
When $C$ is not a proper subset of an existing coalition, $\forall i\in C$ we compare pessimistic values for $C$ to optimistic values for $\gamma(i)$:
in the FOHGS case $\min{u_i(C)} = |N|\cdot|F_i\cap C| - |E_i\cap C| - |S_i\cap C|$, $\max{u_i(\gamma(i))} = |N|\cdot|F_i\cap \gamma(i)| - |E_i\cap\gamma(i)| - |S_i\cap C^*| + |N|\cdot|S_i\cap \gamma(i)^*|$,
in the EOHGS case$\min{u_i(C)} = |F_i\cap C| - |N|\cdot|E_i\cap C| - |N|\cdot|S_i\cap C|$, $\max{u_i(\gamma(i))} = |F_i\cap \gamma(i)| - |N|\cdot|E_i\cap\gamma(i)| - |N|\cdot|S_i\cap C^*| + |S_i\cap \gamma(i)^*|$.
Thus, $C$ can be verified to necessarily block $\gamma$ in polynomial time for FOHGS and EOHGS.
\qed
\end{proof}

We can also check whether $C$ weakly blocks $\gamma$ in polynomial time.
\begin{proposition}
\label{prop:FEOHGS_PSCS_incoNP}
A coalition $C$ can be verified to possibly or necessarily weakly block a partition $\gamma$ for a FOHGS or EOHGS instance $G$ in P-time.
\end{proposition}

\begin{proof}
To check whether $C$ weakly possibly or necessarily blocks $\gamma$, we can extend the logic used in Theorem \ref{thm:FOHGS_OnlyExtremes} and Proposition \ref{prop:FEOHGS_PN_blocking}.
Instead of ensuring that \textit{all} agents are strictly better off in a candidate coalition $C$, we check whether all $i\in C$ \textit{weakly} benefit from deviating, then use a second check to make sure that at least \textit{one} agent strictly benefits from the deviation.
Thus, we can determine whether $C$ possibly or necessarily weakly blocks a partition $\gamma$ in P-time.
\qed
\end{proof}

Since possibly and necessarily (weakly) blocking coalitions can be verified in P-time, we can better understand the hardness of N-CS and N-SCS verification for FOHGS and EOHGS.
Proposition \ref{prop:FEOHGS_NCoreVerif_coNPC} builds on the above results in addition to results from Chen et al. and Sung and Dimitrov \cite{chenhedonic23,SUNG2007155}.



\begin{proposition}
\label{prop:FEOHGS_NCoreVerif_coNPC}
Verification of N-CS and N-SCS for FOHGS EOHGS are coNP-complete.
\end{proposition}

\begin{proof}
It follows from Theorem \ref{thm:FOHGS_OnlyExtremes}, Proposition \ref{prop:FEOHGS_PN_blocking}, and Proposition \ref{prop:FEOHGS_PSCS_incoNP} that the CS and SCS verification problems for FOHGS and EOHGS are contained in coNP.

Sung and Dimitrov proved that CS and SCS verification are coNP-complete for EOHG \cite{SUNG2007155}.
More recently, Chen et al. proved that CS and SCS verification are coNP-complete for FOGH and EOHG \cite{chenhedonic23}.
Since CS and SCS verification is coNP-complete for FOHG and EOGH, and FOHG (resp. EOHG) special cases of FOHGS (resp. EOHGS), it follows that N-CS and N-SCS verification are coNP-hard for both FOHGS and EOHGS.
Thus N-CS and N-SCS verification are coNP-hard and contained in coNP, meaning they are coNP-complete. \qed
\end{proof}

Our next result shows that P-SCS, P-CS, and 
P-INS verification are in co-NP.
Note that this result cannot be directly implied 
by Proposition~\ref{prop:FEOHGS_PN_blocking}.
\begin{proposition}
\label{prop:FEOHGS_PCoreVerif_coNP}
P-SCS, P-CS, and P-INS verification are in coNP for FOHGS and EOHGS.
\end{proposition}
One may assume that P-SCS, P-CS, and P-INS can be verified by showing that no necessarily (weakly) blocking coalition exists, but this  assumption is flawed.
Even if no necessarily (weakly) blocking coalition 
exists, there might exist a set of coalitions
$\mathbb{C}=\{C_1, ..., C_k\}$, 
s.t. for each resolution of $S$, 
$\exists C\in \mathbb{C}$ where 
$\gamma$ is blocked by $C$, i.e., 
it might be the case that no single coalition
blocks for all possible resolutions, but 
multiple coalitions cover all possible resolutions.
If the size of $\mathbb{C}$ becomes exponential, 
then P-CS verification can be 
a $\Sigma_2^P$ problem.
The same logic holds for P-SCS and P-INS verification.
The full proof, found in Appendix A, conclusively proves that a necessary (weakly) blocking coalition must exist if $\gamma$ is not P-SCS, P-CS, or P-INS.
Next, we provide possible and necessary verification results for individual-based stability notions.



\begin{proposition}
\label{prop:FEOHGS_PNNashS_P}
P-NS and N-NS for FOHGS and EOHGS can be verified in P-time.
\end{proposition}


\begin{proof}
We define the resolution $r'$ s.t.,
given instance $G=(N,F,E,S)$ and partition $\gamma$, 
each edge $(i,j)\in S$ becomes friendship if $\gamma(i)=\gamma(j)$, and otherwise becomes enmity.
We observe that under $r'$, $\forall i\in N$, $u_i^{r'}(\gamma(i))=u_i^{r^+}(\gamma(i))$ and $\forall C\in \gamma: C\neq \gamma(i)$, 
$u_i^{r'}(C\cup\{i\})=u_i^{r^-}(C\cup\{i\}$.
It follows that if $\nexists i\in N$ 
and $C\in \gamma: C\neq \gamma(i)$ s.t. $u_i^{r'}(C\cup\{i\}) > u_i^{r'}(\gamma(i))$, then $\gamma$ is P-NS.

Next, we define the resolution $r^*$ s.t.,
given instance $G=(N,F,E,S)$ and partition $\gamma$
each edge $(i,j)\in S$ becomes enmity if $\gamma(i)=\gamma(j)$, and otherwise becomes friendship.
Next, we observe that under $r^*$, $\forall i\in N$, $u_i^{r^*}(\gamma(i))=u_i^{r^-}(\gamma(i))$ and $\forall C\in \gamma: C\neq \gamma(i)$, 
$u_i^{r^*}(C\cup\{i\})=u_i^{r^+}(C\cup\{i\})$.
It reasonably follows that if $\nexists i\in N$ and 
$C\in \gamma: C\neq \gamma(i)$ s.t. $u_i^{r^*}(C\cup\{i\}]) > u_i^{r^*}(\gamma(i))$, then $\gamma$ is N-NS.
\qed
\end{proof}




Our results P-IS, N-IS, P-CIS, N-CIS, P-IR, and N-IR verification follow similar logic to our N-NS and P-NS verification results.

\begin{theorem}
\label{thm:FEOHGS_PNCISIS_P}
P-IS, N-IS, P-CIS, N-CIS, P-IR, and N-IR for FOHGS and EOHGS can be verified in P-time.
\end{theorem}

\begin{proof}
We rely on Lemma \ref{lem:FOHGS_MinMax} to simplify the number of resolutions of $S$ we must consider for some EOHGS or FOHGS $G$.
If some partition $\gamma$ is P-CIS, then there must exist no $i\in N$ and $C\in \gamma: C\neq \gamma(i)$ which satisfy the following conditions: $u_i^{r^-}(C\cup\{i\}) > u_i^{r^+}(\gamma(i))$ (necessary desire to move), $\forall j\in C$ $i\in F_j$ (necessary permission to enter), and $\forall k\in \gamma(i) \setminus\{i\}$ $i\in E_k$ (necessary permission to leave).
Verifying P-IS only requires that there exists no $i\in N$ and $C\in \gamma: C\neq \gamma(i)$ which satisfy the first two of these three conditions.
Since $|\{C\subseteq N: C\in \gamma\}| \leq |N|$, we can conclude that P-CIS and P-IS can be verified in P-time.

If some partition $\gamma$ is N-CIS, then there must exist no $i\in N$ and $C\in \gamma: C\neq \gamma(i)$ which satisfy the following conditions: $u_i^{r^+}(C\cup\{i\}) > u_i^{r^-}(\gamma(i))$ (possible desire to move), $\forall j\in C$ $i\in F_j$ or $i\in S_j$ (possible permission to enter), and $\forall k\in \gamma(i) \setminus\{i\}$ $i\in E_k \cup S_k$ (possible permission to leave).
Verifying N-IS only requires that there exists no $i\in N$ and $C\in \gamma: C\neq \gamma(i)$ which satisfy the first two of these three conditions.
Since $|\{C\subseteq N: C\in \gamma\}| \leq |N|$, we can conclude that N-CIS and N-IS can be verified in P-time.

If some partition $\gamma$ is P-IR, then there must exist no $i\in N$ such that $u_i^{r^+}(\gamma(i)) < 0$.
If some partition $\gamma$ is N-IR, then there must exist no $i\in N$ such that $u_i^{r^-}(\gamma(i)) < 0$.
Thus P-IR and N-IR can be verified in P-time.
\qed
\end{proof}


\section{Possible and Necessary Stability Existence}
We first observe that every FOHGS and EOHGS is identical to some FOHG or EOHG respectively once all strangers have been resolved into either friends or enemies.
This leads to the following straightforward observation, which begets several trivial results, which we omit for brevity.

\begin{observation}
\label{obs:FEOHG_PEXIST}
If X-stable partitions are guaranteed to exist for FOHG (resp. EOHG), then possibly X-stable partitions must exist for FOHGS (resp. EOHGS).
\end{observation}





Answering the question of whether necessarily stable partitions exist is more involved than the question of possible stability existence.
Note that all relevant graphs use black lines between nodes to denote known friends and dashed blue lines to denote strangers.





\begin{theorem}
\label{thm:FEOHGS_NoCIS}
N-CIS partitions are not guaranteed to exist for FOHGS and EOHGS even when relations are symmetric.
\end{theorem}

\begin{proof}
Consider a FOHGS or EOHGS instance with two agents $1$ and $2$ who are mutual strangers with each other.
Since there are only two agents, there are two possible partitions for this game, which we denote as: $\gamma_1=\{\{1\}, \{2\}\}$ and $\gamma_2=\{\{1,2\}\}$; both partitions give both agents a utility of 0 before the stranger edge is resolved.
If the stranger edge resolves to friendship, then $\gamma_1$ would be CIS and $\gamma_2$ would not be. If the stranger edge resolves to enmity, then $\gamma_2$ would be CIS and $\gamma_1$ would not be. Thus, regardless of how the stranger edge is resolved, neither partition is N-CIS
Thus N-CIS partitions are not guaranteed to exist for symmetric FOHGS and EOHGS.
\qed
\end{proof}

\begin{theorem}
\label{thm:EOHGS_NoIndivSym}
N-IS partitions are not guaranteed to exist for EOHGS, even when relations are symmetric and all agents have known friends.
\end{theorem}
\begin{proof}
Consider a symmetric EOHGS with three agents:, $\{1, 2, 3\}$, symmetric friend edges $(1,2), (1,3)$, stranger edge $(2,3)$, and no enemy edges.
With three agents we can exhaustively analyze each possible partition of the agents.
These are: 
(i) $\{\{1\}, \{2\}, \{3\}\}$,
(ii) $\{\{1,2\}, \{3\}\}$, 
(iii) $\{\{1,3\}, \{2\}\}$, 
(iv) $\{\{1\}, \{2,3\}\}$, and 
(v) $\{\{1,2,3\}\}$.

All agents have a utility of 0 in partition (i). It is not N-IS, since both 1 and 2 (or 1 and 3) 
desire to move to (ii) (or to (iii)) regardless of the outcome of $(2,3)$.
Agents 1 and 2 have a utility of 1, while agent 3 has a utility of 0 in partition (ii).
If $(2,3)$ becomes friendship, 
all agents desire to move to (v). Thus, partition (ii) is not N-IS.
Agents 1 and 3 have a utility of 1, while agent 2 has a utility of 0 in partition (iii). 
If $(2,3)$ becomes friendship, 
all agents desire to move to (v). Thus, partition (iii) is not N-IS.
Agent 1 has a utility of 0 and agents 2 and 3 have a utility of $-3$ or $1$ depending on how $(2,3)$ is resolved in partition (vi).
If $(2,3)$ becomes friendship, 
all agents desire to move to (v). Thus, partition (vi) is not N-IS.
Agent 1 has a utility of 2 and agents 2 and 3 both have utilities of either 2 or -1 in partition (v) depending on how $(2,3)$ is resolved.
If $(2,3)$ becomes enmity, then agent 2 (or agent 3) desires to move to 
(iii) (or to (ii)). 
Thus, partition (v) is not N-IS.
Since none of the possible partitions of $G$ are N-IS, we conclude that N-IS partitions are not guaranteed to exist for EOHGS, even when relations are symmetric and all agents have known friends.
\qed
\end{proof}





Corollary \ref{cor:FEOHGS_No_NIS_NNS_NSCS} follows from Theorem \ref{thm:FEOHGS_NoCIS}.

\begin{corollary}
\label{cor:FEOHGS_No_NIS_NNS_NSCS}
N-IS, N-NS, and N-SCS partitions are not guaranteed to exist even when relations are symmetric.
\end{corollary}

We introduce a restriction that guarantees the existence of N-CIS partitions.



\begin{proposition}
\label{prop:FEOHGS_CISFriendGuarantee}
If all agents are viewed as a known friend by at least one agent, then N-CIS partitions must exist for FOHGS and EOHGS.
\end{proposition}

\begin{proof}
Consider a FOHGS or EOHGS where each agent is viewed as a friend by at least one agent and let the grand coalition form.
Now suppose that there exists some agent $i\in N$ who is motivated to become a singleton for some resolution of strangers; in the FOHGS case, this means that $F_i = \emptyset$, while in the EOHGS case, it means that $E_i\neq \emptyset$.
Agent $i$ does not require permission to join their destination coalition, because they want to leave to become a singleton.
$\exists k\in N$ such that $i\in F_k$, so $i$ cannot get permission to leave the coalition regardless of how stranger relations are resolved.
Thus, the grand coalition is N-CIS.
\qed
\end{proof}

While ensuring that each agent is viewed as a friend by at least one other agent is sufficient to guarantee the existence of N-CIS partitions, it is not a strictly necessary requirement.

\begin{proposition}
\label{prop:FEOHGS_FriendSufNotNecCIS}
Ensuring that each agent is viewed as a friend by at least one agent is not a necessary condition for N-CIS in FOHGS or EOHGS.
\end{proposition}

\begin{proof}
A trivial case with no strangers is a FOHG or EOHG instance for which a CIS (and thus N-CIS) partition must exist. The claim also holds for other cases.

Consider a symmetric FOHGS or EOHGS $G$ with five agents $\{1,2,3,4,5\}$ with undirected known friend edges $(1,2),(4,5)$ and stranger edges $(2,3),(3,4)$.
The partition $\{\{1,2\},\{3\},\{4,5\}\}$ is N-CIS, because only 3 wants to deviate, but 1 denies 3 permission to join $\{1,2\}$ if 2 and 3 become friends, and 5 denies permission to join $\{4,5\}$ if 3 and 4 become friends.
Further, in the EOHGS case no agent can benefit by deviating.
Thus, no matter how the stranger edges are resolved no agent has motivation and permission to deviate, so the partition $\{\{1,2\},\{3\},\{4,5\}\}$ is N-CIS.
Thus N-CIS partitions can exist for FOHGS and EOHGS where there exists an agent who is not viewed as a friend by anyone.
\end{proof}





Since NS implies CIS, Corollary \ref{cor:EOHGS_noNash} follows from Theorem \ref{thm:EOHGS_NoIndivSym}.


\begin{corollary}
\label{cor:EOHGS_noNash}
N-NS partitions are not guaranteed to exist for EOHGS instances, even when relations are symmetric and all agents have known friends.
\end{corollary}

We now provide a positive result for symmetric FOHGS.

\begin{theorem}
\label{thm:Symm_FOHGS_NNS_PT}
N-NS existence can be checked for symmetric FOHGS in P-time.
\end{theorem}


In the full proof for Theorem \ref{thm:Symm_FOHGS_NNS_PT}, we provide a method to check N-NS existence for a symmetric FOHGS instance in P-time.
Since N-NS verification is in P (Prop. \ref{prop:FEOHGS_PNNashS_P}), asymmetric FOHGS and EOHGS model asymmetric FOHG and EOHG for which and NS existence is NP-complete \cite{brandt2022single}, we make Observation \ref{obs:FEOHGS_NNS_NPC}.
\begin{observation}
    \label{obs:FEOHGS_NNS_NPC}
    N-NS existence is NP-complete for asymmetric FOHGS and EOHGS.
\end{observation}

\cite{brandt2022single} proved that IS existence can be checked for FOHG and EOHG in P-time. 
We prove that adding strangers makes the problem NP-complete for asymmetric FOHGS and EOHGS.


\begin{theorem}
\label{thm:Asymm_FOHGS_NIS_NPC}
N-IS existence is NP-complete for asymmetric FOHGS and EOHGS.
\end{theorem}


The proof for Theorem \ref{thm:Asymm_FOHGS_NIS_NPC} utilizes a reduction from exact cover by 3 sets (EC3) that builds on existing work by 
\cite{brandt2022single}, but defer the full proof to Appendix B.
We now examine N-CIS existence.

\begin{theorem}
\label{thm:Symm_FOHGS_NCIS_PT}
N-CIS existence can be checked for symmetric FOHGS and symmetric EOHGS in P-time.
\end{theorem}

We provide a deterministic method to check N-CIS existence in P-time in the full proof for Theorem \ref{thm:Symm_FOHGS_NCIS_PT}; see Appendix B.
Next, we show that N-CIS is NP-complete for asymmetric FOHGS.

\begin{theorem}
\label{thm:Asymm_FOHGS_NCIS_NPC}
N-CIS existence is NP-complete for asymmetric FOHGS.
\end{theorem}


The proof for Theorem \ref{thm:Asymm_FOHGS_NCIS_NPC} also uses a reduction from EC3, but is substantially more involved than the proof for Theorem \ref{thm:Asymm_FOHGS_NIS_NPC}, see Appendix B.
Our next result shows that even fairly strict graph restrictions cannot guarantee N-CS.

\begin{figure}
    \centering
    \begin{tabular}{cc}
        \includegraphics[width=0.25\linewidth]{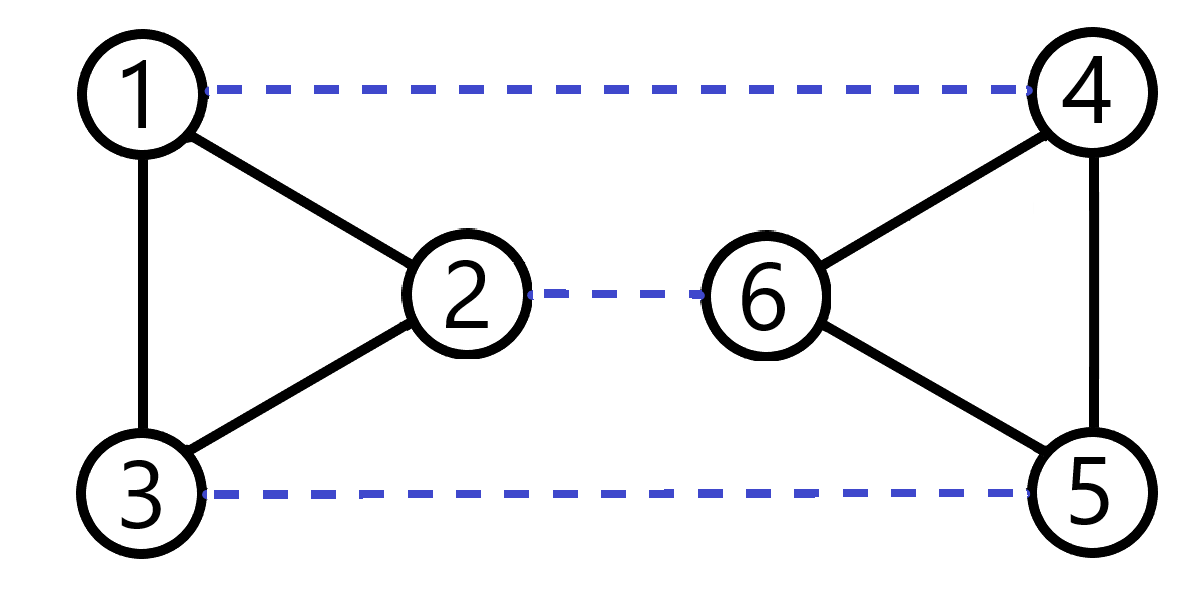} & \includegraphics[width=0.25\linewidth]{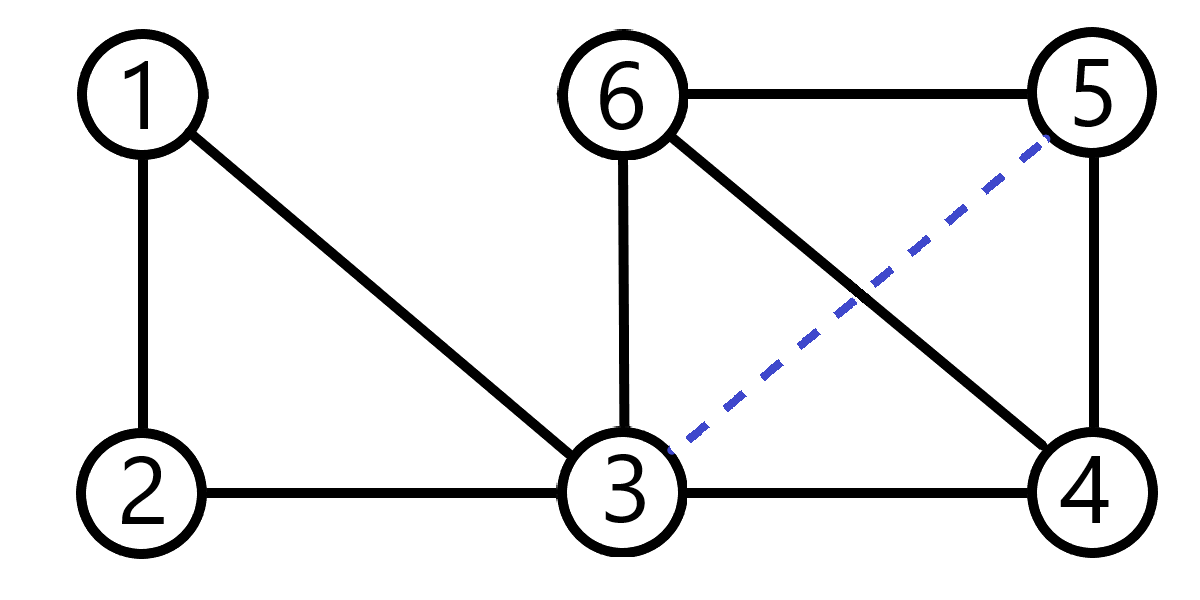} \\
        (a) & (b) \\
    \end{tabular}
    \caption{No N-CS partition (a) FOHGS, (b) EOHGS}
    \label{fig:NoNecessarySCS}
\end{figure}



\begin{theorem}
\label{thm:SymmFOHGS_NoCoreS}
N-CS partitions are not guaranteed to exist for FOHGS, even when relations are symmetric, each agent has more known friends than strangers, and each agent has at most one stranger.
\end{theorem}


We show that a FOHGS based on Figure \ref{fig:NoNecessarySCS} (a) has no N-CS partition, see Appendix B.
Since SCS implies CS, Corollary \ref{cor:SymmFOHGS_NoNSCS} follows from Theorem \ref{thm:SymmFOHGS_NoCoreS}.

\begin{corollary}
\label{cor:SymmFOHGS_NoNSCS}
N-SCS partitions are not guaranteed to exist for symmetric FOHGS.
\end{corollary}



Theorem \ref{thm:EOHGS_SymmCondNoGuarantee} shows that N-CS partitions may not exist for EOHGS under more restrictive conditions.

\begin{theorem}
\label{thm:EOHGS_SymmCondNoGuarantee}
N-CS partitions are not guaranteed to exist for symmetric EOHGS, agents have more known friends than strangers, at most one stranger, and only one pair of agents are strangers.
\end{theorem}




We show that an EOHGS based on Figure \ref{fig:NoNecessarySCS} (b) has no N-CS partition, see Appendix B.
While Theorems \ref{thm:SymmFOHGS_NoCoreS} and \ref{thm:EOHGS_SymmCondNoGuarantee} show similar negative results, the counterexamples used in their proofs are not interchangeable.
While N-CS cannot be guaranteed, even under relatively strong graph restrictions, we now present a positive result for N-INS existence.

\begin{proposition}
\label{prop:FOHGS_SCCNInternalStab}
A partition of strongly connected components on the friendship graph is always N-INS for FOHGS.
\end{proposition}

\begin{proof}
Suppose that a partition of strongly connected components on the friendship graph $\gamma$ for some symmetric FOHGS $G$ was not N-INS.
There must $\exists C\in \gamma$ and some $D\subset C$ such that $\forall i\in D$, $\max{u_i(D)} > \min{u_i(C)}$ and $\min{u_i(D)} > \min{u_i(C)}$.
In order for $\min{u_i(D)} > \min{u_i(C)}$ to hold $\forall i\in D$, it must be the case that $\forall i\in D$, $F_i \cap (C\setminus\{D\}) = \emptyset$ else there would exist some $i$ such that $\min{u_i(D)} < \min{u_i(C)}$.
Thus, the condition $\forall i\in D$, $\min{u_i(D)} > \min{u_i(C)}$ can only hold when $D$ is possibly disjoint from $C$ on the union of the friendship and stranger graph, which contradicts the basic premise that all $C\in \gamma$ are strongly connected components on the friendship graph.
Thus a partition of strongly connected components on the friendship graph is N-INS for FOHGS.
\qed
\end{proof}





\begin{proposition}
\label{prop:EOHGS_CPNInternalStab}
A partition of cliques on the friendship graph is always N-INS for EOHGS.
\end{proposition}

\begin{proof}
In a directed graph, a clique is a subset $C\subseteq N$ such that $\forall i,j\in C$ the edges $(i,j)$ and $(j,i)$ exist.

Suppose for some EOHGS $G$ there is a partition of cliques $\gamma$ that is not N-INS.
There must $\exists C\in \gamma$ and some $D\subset C$ such that, for some resolution of strangers, $\forall i\in D$, $i$ is possibly better off in $D$ than in $C$.
Since all $C\in \gamma$ are cliques on the friendship graph, we know that $E_i\cap D = E_i\cap C = \emptyset$, $S_i\cap C = \emptyset$, and $|F_i\cap D| < |F_i\cap C|$.
Thus, $\forall C\in \gamma$, $\forall D\subset C$, $\forall i\in D$, $u_i(D) < u_i(C)$ necessarily holds, which contradicts the claim that there exists some $C\in \gamma$, $D\subset C$, such that $\forall i\in D$, $u_i(D) > u_i(C)$ possibly holds.
Thus a partition of cliques on the friendship graph is always N-INS for EOHGS.
\qed
\end{proof}

\section{Conclusions and Future Work}
We introduced friend-oriented hedonic games with strangers (FOHGS) and enemy-oriented hedonic games with strangers (EOHGS), new extensions to the friend- and enemy-oriented hedonic games models (FOHG and EOHG respectively).
We showed that possible and necessary stability verification for FOHGS and EOHGS is in the same complexity class as stability verification for FOHG and EOHG for several notions of stability.
We outlined straightforward possible stability guarantees for FOHGS and EOHGS.
We provided several results characterizing the hardness of checking necessary stability existence for several stability notions, but leave additional characterization for future work.
Future work may also include the development of heuristics to verify group-based notions of possible and necessary stability (e.g. N-CS) in FOHGS and EOHGS as well as heuristics to check the existence of individual-based notions of necessary stability (e.g. N-NS) in asymmetric FOHGS and EOHGS.



%
%
%
\bibliographystyle{splncs04}
\bibliography{prima_sub}

\begin{thebibliography}{10}
\providecommand{\url}[1]{\texttt{#1}}
\providecommand{\urlprefix}{URL }
\providecommand{\doi}[1]{https://doi.org/#1}

\bibitem{alcalde2006coalition}
Alcalde, J., Romero-Medina, A.: Coalition formation and stability. Social
  Choice and Welfare  \textbf{27}(2),  365--375 (2006)

\bibitem{banerjee2001core}
Banerjee, S., Konishi, H., S{\"o}nmez, T.: Core in a simple coalition formation
  game. Social Choice and Welfare  \textbf{18}(1),  135--153 (2001)

\bibitem{barrot2019unknown}
Barrot, N., Ota, K., Sakurai, Y., Yokoo, M.: Unknown agents in friends oriented
  hedonic games: Stability and complexity. In: Proc. 33rd AAAI Conference on
  Artificial Intelligence. vol.~33, pp. 1756--1763 (2019)

\bibitem{bogomolnaia2002stability}
Bogomolnaia, A., Jackson, M.O.: The stability of hedonic coalition structures.
  Games and Economic Behavior  \textbf{38}(2),  201--230 (2002)

\bibitem{brandt2022single}
Brandt, F., Bullinger, M., Tappe, L.: Single-agent dynamics in additively
  separable hedonic games. In: Proc. 36th AAAI Conference on Artificial
  Intelligence. pp. 4867--4874 (2022)

\bibitem{chenhedonic23}
Chen, J., Cs{\'a}ji, G., Roy, S., Simola, S.: Hedonic games with friends,
  enemies, and neutrals: Resolving open questions and fine-grained complexity.
  In: Proc. 22nd International Conference on Autonomous Agents and Multiagent
  Systems. pp. 251--259 (2023)

\bibitem{dimitrov2006simple}
Dimitrov, D., Borm, P., Hendrickx, R., Sung, S.C.: Simple priorities and core
  stability in hedonic games. Social Choice and Welfare  \textbf{26}(2),
  421--433 (2006)

\bibitem{DrezeGreenbergHC1980}
Drèze, J.H., Greenberg, J.: Hedonic coalitions: Optimality and stability.
  Econometrica  \textbf{48}(4),  987--1003 (1980)

\bibitem{rothe2015economics_Chapter3}
Elkind, E., Rothe, J.: Economics and Computation: An Introduction to
  Algorithmic Game Theory, Computational Social Choice, and Fair Division,
  chap.~3. Springer Texts in Business and Economics (2016)

\bibitem{kerkmann2020hedonic}
Kerkmann, A.M., Lang, J., Rey, A., Rothe, J., Schadrack, H., Schend, L.:
  Hedonic games with ordinal preferences and thresholds. Journal of Artificial
  Intelligence Research  \textbf{67},  705--756 (2020)

\bibitem{kerkmann2020altruism}
Kerkmann, A.M., Rothe, J.: Altruism in coalition formation games. In: Proc.
  29th International Joint Conferences on Artificial Intelligence. pp. 347--353
  (2020)

\bibitem{konczak2005voting}
Konczak, K., Lang, J.: Voting procedures with incomplete preferences. In: Proc.
  14th International Joint Conferences on Artificial Intelligence
  Multidisciplinary Workshop on Advances in Preference Handling. vol.~20 (2005)

\bibitem{lang2015representing}
Lang, J., Rey, A., Rothe, J., Schadrack, H., Schend, L.: Representing and
  solving hedonic games with ordinal preferences and thresholds. In: Proc. 14th
  International Conference on Autonomous Agents and Multiagent Systems. pp.
  1229--1237 (2015)

\bibitem{nguyen2016altruistic}
Nguyen, N.T., Rey, A., Rey, L., Rothe, J., Schend, L.: Altruistic hedonic
  games. In: Proc. 15th International Conference on Autonomous Agents and
  Multiagent Systems. pp. 251--259 (2016)

\bibitem{ota2017core}
Ota, K., Barrot, N., Ismaili, A., Sakurai, Y., Yokoo, M.: Core stability in
  hedonic games among friends and enemies: Impact of neutrals. In: Proc. 26th
  International Joint Conferences on Artificial Intelligence. pp. 359--365
  (2017)

\bibitem{peters2015simple}
Peters, D., Elkind, E.: Simple causes of complexity in hedonic games. In: Proc.
  24th International Joint Conferences on Artificial Intelligence (2015)

\bibitem{rey2016beyond}
Rey, A.: Beyond Intractability: A Computational Complexity Analysis of Various
  Types of Influence and Stability in Cooperative Games. Ph.D. thesis,
  Heinrich-Heine-Universität Düsseldorf (2016)

\bibitem{schlueter2020internal}
Schlueter, J., Goldsmith, J.: Internal stability in hedonic games. In: Proc.
  33rd International FLAIRS Conference. pp. 160--165 (2020)

\bibitem{schlueter2020super}
Schlueter, J., Goldsmith, J.: Super altruistic hedonic games. In: Proc. 33rd
  International FLAIRS Conference. pp. 154--159 (2020)

\bibitem{skibski2020signed}
Skibski, O., Suzuki, T., Grabowski, T., Michalak, T., Yokoo, M.: Signed graph
  games: Coalitional games with friends, enemies and allies. In: Proc. 19th
  International Conference on Autonomous Agents and Multiagent Systems. pp.
  1287--1295 (2020)

\bibitem{SUNG2007155}
Sung, S.C., Dimitrov, D.: On core membership testing for hedonic coalition
  formation games. Operations Research Letters  \textbf{35}(2),  155 -- 158
  (2007)

\bibitem{sung2010computational}
Sung, S.C., Dimitrov, D.: Computational complexity in additive hedonic games.
  European Journal of Operational Research  \textbf{203}(3),  635--639 (2010)

\bibitem{taywade2020decentralized}
Taywade, K., Goldsmith, J., Harrison, B.: Decentralized marriage models. In:
  Proc. 33rd International FLAIRS Conference. pp. 213--216 (2020)

\end{thebibliography}
%





\appendix

\section{Appendix A: Section 4 Proofs}


Proof for Lemma \ref{lem:FOHGS_MinMax}.

\begin{proof}
Suppose that $r^-$, does not provide the absolute minimum utility for agents evaluating $\gamma$. When every edge $(i,j) \in S$ is resolved, it must become either friendship or enmity, meaning that any resolution where not all edges become enemy relationships must contain at least one edge which becomes friendship instead.
Consider some pair of agents $i,k\in N\times N$ such that $k\in S_i$.
Let $r_k$ be some resolution where $(i,k)$ and becomes friendship and all other stranger relations become enmity.
Now suppose that $u_i^{r^-}(\gamma(i)) > u_i^{r_k}(\gamma(i))$.
Consider the utility $i$ derives from $\gamma$ under both proposed resolutions.
In both cases, we begin by determining $|F_i\cap \gamma(i)|$ then multiplying that by $|N|$; next subtract $|E_i\cap \gamma(i)|$ from the result, so to determine whether $u_i^{r^-}(\gamma(i)) > u_i^{r_k}(\gamma(i))$ we can subtract $(|N|\cdot|F_i\cap \gamma(i)| - |E_i\cap \gamma(i)|)$ from both sides and focus on the utility derived from stranger edges.

Given that the utility an agent derives from stranger edges is defined by $\sum_{j\in S_i\cap \gamma(i)} v_i(j)$ and that all strangers become enemies under $r^-$, we can simplify the utility equation to $\sum_{j\in S_i\cap \gamma(i)} -1 = -|S_i\cap \gamma(i)|$ for $r^-$.
The minimum possible value for this equation is $-(n-1)$ when $S_i = \{j:j\in N\setminus \{i\}\}$ and $\gamma(i)$ is the grand coalition.
When computing the utility $i$ derives from stranger edges under $r_k$, we must keep in mind that $v_i(k) = |N|$. If $k\in \gamma(i)$, then $\sum_{j\in S_i\cap \gamma(i)} v_i(j) = |N| - (|S_i\cap \gamma(i)| - 1)$ under $r_k$, which is greater than the $-|S_i\cap \gamma(i)|$ obtained under $r^-$.
If, instead $k\notin \gamma(i)$, then $\sum_{j\in S_i\cap \gamma(i)} -1 = -|S_i\cap \gamma(i)|$ under $r_k$; this is equivalent to the value obtained under $r^-$, but critically is \textit{not} less than the value obtained under $r^-$.
It follows that, for any other $j\in S_i\setminus \{k\}$, treating $j$ as a friend cannot decrease the utility that $i$ derives from $\gamma(i)$.
Thus $u_i^{r^-}(\gamma(i))$ is always the minimum possible value $\forall i\in N$.

Next consider the resolution where all edges in $S$ become friendship, $r^+$. We now prove that $r^+$ always produces the maximum possible value for $u_i^{r^+}(\gamma(i)) \mbox{ } \forall i\in N$. Our proof follows the same general setup as the proof that $u_i^{r^-}(\gamma(i))$ minimizes the utility of some agent $i$.
We consider some resolution $r_k$ under which there exists a pair of agents $i,k\in N\times N$ such that $(i,k)$ becomes enmity and all other stranger relations become friendship.
The value of $\sum_{j\in S_i\cap \gamma(i)} v_i(j)$ under $r^+$ is $\sum_{j\in S_i\cap \gamma(i)} |N| = |N|\cdot |S_i\cap \gamma(i)|$.
Under $r_k$, the value of $\sum_{j\in S_i\cap \gamma(i)} v_i(j)$ depends on whether $k\in \gamma(i)$ or not.
If $k\in \gamma(i)$ then $\sum_{j\in S_i\cap \gamma(i)} v_i(j) = (|N|-1)\cdot |S_i\cap \gamma(i)| - 1$ which is strictly less than $|N|\cdot |S_i\cap \gamma(i)|$.
If $k\notin \gamma(i)$ then $\sum_{j\in S_i\cap \gamma(i)} v_i(j) = |N|\cdot |S_i\cap \gamma(i)|$ which is equal to the value obtained under $r^+$.
It follows that treating any other $j\in S_i\setminus\{k\}$ as an enemy instead of a friend can have no positive impact on $u_i(\gamma(i))$, so we conclude that $u_i^{r^+}(\gamma(i))$ is the maximum utility $\forall i\in N$.

An analogous argument holds for EOHGS where the magnitudes of the positive and negative values are exchanged.
\end{proof}

\setcounter{proposition}{4} 



Proof for Proposition \ref{prop:FEOHGS_PCoreVerif_coNP}.
\begin{proof}
    We begin by defining the resolution $r'$ s.t.,
    given FOHGS or EOHGS instance $G=(N,F,E,S)$ and partition of $N$, $\gamma$,
    each edge $(i,j)\in S$ becomes friendship if $\gamma(i)=\gamma(j)$, and otherwise becomes enmity.
    Note that this construction guarantees that $\forall i\in N$, $u_i^{r'}(\gamma(i))=u_i^{r^+}(\gamma(i))$.


    To construct $r'$, we first observe that since $\bigcup_{C\in \gamma}C=N$, it's possible to traverse $\forall C\in \gamma$, $\forall i\in C$ in $O(|N|)$ time, so the nested traversal is roughly equivalent to $\forall i\in N$.
    Recall that $\forall i\in N$, there are $|N|-1$ agents $j\in N: j\neq i$, so $|S_i|\leq |N|-1$.
    Given a coalition $C$ and some agent $j\in S_i$, we can determine whether $j\in C$ in $O(|N|)$ time since $\forall C\in \gamma$, $|C|\leq |N|$.
    Assigning a value to edge $(i,j)\in r'$ is assumed to be a $O(1)$ operation, given prior knowledge of $i$ and $j$.
    Thus, $r'$ can be computed in polynomial (specifically $O(|N|^2)$) time.

    Having defined $r'$ and shown that it can be computed efficiently, we now observe that if no coalition possibly blocks $\gamma$ under $r'$, then $\gamma$ is, by definition, P-CS; this immediately implies that $\gamma$ is also P-INS.
    Using a similar argument, if no coalition possibly weakly blocks $\gamma$ under $r'$, then $\gamma$ is P-SCS.
    This argument holds for both FOHGS and EOHGS.
    
    Now suppose that $\gamma$ is not P-CS.
    It may be reasonable to assume that this means that some necessary blocking coalition $C$ exists. It is certainly true that the existence of a necessarily blocking coalition is sufficient to disprove the claim that $\gamma$ is P-CS.
    However, while the condition is sufficient, it does not immediately follow that it is necessary and sufficient.
    
    
    Suppose that, instead of being necessarily blocked by a single coalition $C$, $\gamma$ is instead possibly blocked by a set of coalitions in such a way that no resolution exists where $\gamma$ is P-CS.
    More formally, suppose that $\exists\mathbb{C}=\{C_1, ..., C_k\}: \forall C\in \mathbb{C}, C\subseteq N$ such that for every resolution of $S$, $\gamma$ is blocked by at least one $C\in \mathbb{C}$, but $\nexists C\in \mathbb{C}$ that necessarily blocks $\gamma$.
    If such a scenario is possible, then P-CS verification can be a $\Sigma_2^P$ problem instead of a coNP problem.

    In this proof, we show that the first supposition, that $\gamma$ is P-CS if no necessarily blocking coalition exists, is correct.
    This is shown by proving that any coalition $C\subseteq N$ that blocks $\gamma$ under $r'$ also necessarily blocks $\gamma$.
    An immediate result of this is that some set $\mathbb{C}$, as outlined above, cannot exist, because at least one member coalition $C$ would have to block $\gamma$ under $r'$ and, as we will soon see, that $C$ would be a necessarily blocking coalition.

    Given some $C$ that possibly (weakly) blocks $\gamma$ under $r'$, we can show that $C$ (weakly) blocks $\gamma$ under any resolution by repeatedly applying the following lemma.

    \begin{lemma}
    \label{lem:rp_edge_flips}
    Assume coalition $C$ (weakly) blocks partition $\gamma$ under $r'$.
    Then $C$ also (weakly) blocks $\gamma$ under a resolution where:
    \begin{enumerate}
        \item one friend edge $(i,j)\in r'$ (thus $\gamma(i)=\gamma(j)$) turns into an an enemy edge, or
        \item one enemy edge $(i,j)\in r'$ (thus $\gamma(i)\neq \gamma(j)$) turns into a friend edge.
    \end{enumerate}
    \end{lemma}
    
    \begin{proof}
        First, we note that the proposed changes only affect $i$ and $j$.
        If neither $i$ nor $j$ is included in $C$, the changes do not matter.
        By symmetry, let us assume $i\in C$.
        Next, we'll check whether $i$ still has an incentive to join $C$ after the above changes.

        For case 1, first assume that $j\in C$.
        Then, for $i$, the difference between the utilities for $C$ and $\gamma(i)$ remains unchanged.
        Thus, $i$ still has an incentive to join $C$ if $C$ blocks $\gamma$ under $r'$.
        Further, $i$ remains at least indifferent if $C$ weakly blocks $\gamma$.
        Next, assume $j\notin C$.
        Then, for $i$, her utility for $C$ does not change, but the utility for $\gamma(i)$ decreases.
        Thus, $i$ still has an incentive to join $C$ if $C$ blocks $\gamma$ under $r'$.
        Further, $i$ gains a strict incentive to join if $C$ weakly blocks $\gamma$ under $r'$.

        For case 2, first assume that $j\in C$.
        Then, for $i$, their utility for $C$ increases, while their utility for $\gamma(i)$ remains unchanged.
        Thus, $i$ still has an incentive to join $C$ if $C$ blocks $\gamma$ under $r'$.
        Further, $i$ gains a strict incentive to join if $C$ weakly blocks $\gamma$ under $r'$.
        Next, assume $j\notin C$.
        Then, for $i$, her utility for $C$ and $\gamma(i)$ remains unchanged.
        Thus, $i$ still has an incentive to join $C$ if $C$ blocks $\gamma$ under $r'$.
        Further, $i$ remains at least indifferent if $C$ weakly blocks $\gamma$.
    \end{proof}

    The findings thus far are sufficient to prove that the P-CS and P-SCS verification problem is in coNP for both FOHGS and EOHGS; however, we still need to prove that this is the case for P-INS.

    \textbf{P-INS Verification}
    
    First, we observe that if $\gamma$ is not P-INS, then $\exists C\in \gamma$ such that $\exists D\subset C$ where $D$ necessarily blocks $C$, and thus also $\gamma$.
    Further, $D$ would also be a necessarily blocking coalition for the purposes of P-CS evaluation.
    Since we already know that P-CS verification is in coNP, it follows that P-INS verification is also contained in coNP.

    Combining the above findings, we conclude that P-SCS, P-CS, and P-INS verification are contained in coNP.
\end{proof}

\section{Appendix B: Section 5 Proof Sketches}
Proof sketches for:
\begin{itemize}
    \item Theorem 5 \ref{thm:Symm_FOHGS_NNS_PT}
    \item Theorem 6 \ref{thm:Asymm_FOHGS_NIS_NPC}
    \item Theorem 7 \ref{thm:Symm_FOHGS_NCIS_PT}
    \item Theorem 8 \ref{thm:Asymm_FOHGS_NCIS_NPC}
    \item Theorem 9 \ref{thm:SymmFOHGS_NoCoreS}
    \item Theorem 10 \ref{thm:EOHGS_SymmCondNoGuarantee}
\end{itemize}

\textbf{Proof sketch for Theorem \ref{thm:Symm_FOHGS_NNS_PT}.}
\begin{proof}
    We begin by introducing a set of necessary and sufficient conditions for an N-NS partition to exist for any symmetric FOHGS, demonstrate that the conditions can be checked in polynomial time, then prove that the conditions are necessary and sufficient.
    \qed
\end{proof}

\textbf{Proof sketch for Theorem \ref{thm:Asymm_FOHGS_NIS_NPC}.}
\begin{proof}
    We construct reductions from Exact Cover by 3 Sets (EC3) similar to those used by 
    Brandt, Bullinger, and Tappe (2022)
    in their proofs that NS existence is NP-complete for FOHG and EOHG. 
    First, we observe that the reductions used by 
    Brandt, Bullinger, and Tappe (2022)
    in the proofs showing the hardness of NS existence can't be used as-is.
    Thus, we modify the reductions used by 
    Brandt, Bullinger, and Tappe (2022)
    to prove the hardness of N-IS existence in FOHGS and EOHGS.
    \qed
\end{proof}

\textbf{Proof sketch for Theorem \ref{thm:Symm_FOHGS_NCIS_PT}.}
\begin{proof}
    We begin by introducing a set of necessary and sufficient conditions for an N-CIS partition to exist for any symmetric FOHGS, briefly observe that the conditions can be checked in polynomial time, then prove that the conditions are necessary and sufficient; conditions being that, for all $i\in N$ one of the following must hold:
    $F_i\neq \emptyset$ OR
    $\forall j\in S_i$, $F_j\neq \emptyset$ AND $\exists k: F_k\neq \emptyset \wedge i\in E_k$
    \qed
\end{proof}

\textbf{Proof sketch for Theorem \ref{thm:Asymm_FOHGS_NCIS_NPC}.}
\begin{proof}
    We construct another reduction from EC3 inspired by proofs by 
    Brandt, Bullinger, and Tappe (2022) and further built off of the proof for Theorem \ref{thm:Asymm_FOHGS_NIS_NPC}.
    \qed
\end{proof}

\textbf{Proof sketch for Theorem \ref{thm:SymmFOHGS_NoCoreS}.}
\begin{proof}
    We show that a FOHGS based on Figure \ref{fig:NoNecessarySCS} (a) has no N-CS partition. In brief, the game is of sufficiently small size that a deterministic analysis of all possible partitions can be conducted, ultimately proving that no N-CS partition exists for the described FOHGS instance.
    \qed
\end{proof}

\textbf{Proof sketch for Theorem \ref{thm:EOHGS_SymmCondNoGuarantee}.}
\begin{proof}
    We show that a FOHGS based on Figure \ref{fig:NoNecessarySCS} (b) has no N-CS partition. Similarly to the proof sketch for Theorem \ref{thm:SymmFOHGS_NoCoreS}, a deterministic analysis of all possible partitions is feasible for an instance of this size. After conducting such an analysis, we find that no N-CS partition exists for the described EOHGS instance.
    \qed
\end{proof}

The full extent of these proofs has been made available via \href{https://www.hostize.com/v/QMyKkT8VI3}{Hostize}.

\end{document}